\documentclass[11pt]{article}

\usepackage{amsmath}
\usepackage{amssymb}

\usepackage{amsthm}
\usepackage{amscd}
\usepackage{amsfonts}
\usepackage{graphicx}
\usepackage{xcolor}
\usepackage{fancyhdr}
\usepackage{url}
\usepackage{comment}
\usepackage{hyperref}
\usepackage{tikz}
\usepackage{setspace}
\usepackage{cite}

\usetikzlibrary{automata,positioning}

\newtheorem{theorem}{Theorem}[section]
\newtheorem{lemma}[theorem]{Lemma}
\newtheorem{proposition}[theorem]{Proposition}
\newtheorem{definition}[theorem]{Definition}

\newcommand{\A}{\ensuremath{\mathop{\text{A}}}}
\newcommand{\E}{\ensuremath{\mathop{\text{E}}}}
\newcommand{\U}{\ensuremath{\mathbin{\text{U}}}}
\newcommand{\R}{\ensuremath{\mathbin{\text{R}}}}

\newcommand{\<}{\langle}
\renewcommand{\>}{\rangle}

\newcommand{\AAEx}{\ensuremath{\mathop{\text{AAE}}}}
\newcommand{\AEAx}{\ensuremath{\mathop{\text{AEA}}}}
\newcommand{\AEx}{\ensuremath{\mathop{\text{AE}}}}

\newcommand{\AAx}{\ensuremath{\mathop{\text{AA}}}}

\newcommand{\X}{\ensuremath{\mathop{\text{X}}}}
\newcommand{\G}{\ensuremath{\mathop{\text{G}}}}
\newcommand{\F}{\ensuremath{\mathop{\text{F}}}}

\newcommand{\Ak}{\mathop{\text{A}\hspace{-1pt}^k}}
\newcommand{\Ej}{\mathop{\text{E}\hspace{1pt}^j}}
\newcommand{\Ek}{\mathop{\text{E}\hspace{0pt}^k}}
\newcommand{\Aj}{\mathop{\text{A}\hspace{0pt}^j}}

\renewcommand{\iff}{\ensuremath{\mathop{\text{iff}}}}

\newcommand{\Atoms}{\ensuremath{\mathsf{Atoms}}}
\newcommand{\atom}[1]{\textsf{#1}}
\newcommand{\Compounds}{\ensuremath{\mathsf{Compounds}}}

\newcommand{\cl}{\mathit{cl}}
\newcommand{\ms}{\mathit{ms}}
\newcommand{\Buchi}{B\"uchi}
\newcommand{\proj}{\mathit{prj}}

\newcommand{\comp}{\ensuremath{\mathit{comp}}}

\newcommand{\zip}{\ensuremath{\mathit{zip}}}
\newcommand{\unzip}{\ensuremath{\mathit{unzip}}}

\newcommand{\Hp}{$\text{HyperLTL}_2$}

\usepackage{hp_defs}

\begin{document}

\title{A temporal logic of security}

\author{
Masoud Koleini \quad Michael R. Clarkson \\
Department of Computer Science \\
George Washington University
 \and  
Kristopher K. Micinski \\
Department of Computer Science \\
University of Maryland, College Park
}

\date{July 9, 2013}

\maketitle

\begin{abstract}
A new logic for verification of security policies is proposed.
The logic, HyperLTL, extends linear-time temporal logic (LTL) with connectives for explicit and simultaneous quantification over multiple execution paths, thereby enabling
HyperLTL to express information-flow security policies that LTL cannot.
A model-checking algorithm for a fragment of HyperLTL is given, and the algorithm is implemented in a prototype model checker.  
The class of security policies expressible in HyperLTL is characterized by an arithmetic hierarchy of hyperproperties.
\end{abstract}


\section{Introduction}
\label{sec:intro}

The theory of \emph{trace properties}, which characterizes correct behavior of programs in terms of properties of individual execution paths, developed out of an interest in proving the correctness of programs~\cite{Lamport77}.
Practical model-checking tools~\cite{NuSMV:2000:Cimatti,Holzmann:1997:SPIN,Hardin:1996:COSPAN,Lamport02} now enable automated verification of correctness.
Verification of security, unfortunately, isn't directly possible with such tools, because
some important security policies require sets of execution paths to model~\cite{McLean96}.
But there is reason to believe that similar verification methodologies could be developed for security:
\begin{itemize}
\item The \emph{self-composition} construction~\cite{BartheDR04,TerauchiA05} reduces properties of pairs of execution paths to properties of single execution paths, thereby enabling verification of a class of security policies.
\item The theory of \emph{hyperproperties}~\cite{ClarksonS10} generalizes the theory of trace properties to security policies, showing that certain classes of security policies are amenable to verification with invariance arguments~\cite{AlpernS87} and with stepwise refinement~\cite{Wirth71}.  
\end{itemize}
Prompted by these ideas, this paper develops an automated verification methodology for security.
In our methodology, security policies are expressed as logical formulas, and a model checker verifies those formulas.

We propose a new logic named \emph{HyperLTL}, which generalizes linear-time temporal logic (LTL)~\cite{Pnueli:1977:TLP}.
LTL implicitly quantifies over only a single execution path of a system, but HyperLTL allows explicit quantification over multiple execution paths simultaneously, as well as propositions that stipulate relationships among those paths.
For example, HyperLTL can express \emph{information-flow policies} such as ``for all execution paths $\pi_1$, there exists an execution path $\pi_2$, such that $\pi_1$ and $\pi_2$ always appear equivalent to observers who are not cleared to view secret information.''
Neither LTL nor branching-time logics (e.g., CTL~\cite{EmersonClarke:1982:CTL} and CTL$^*$~\cite{EmersonH86}) can directly express such policies, because they lack the capability to correlate multiple execution paths~\cite{AlurCZ06,McLean96}.  
Providing that capability is the key idea of HyperLTL.
The syntax that enables it is described in \S\ref{sec:syntax}, along with  several examples of information-flow policies.  
The semantics of HyperLTL is given in \S\ref{sec:semantics}.
It is based on a standard LTL semantics~\cite{Pnueli:1977:TLP} that models a formula with a single \emph{computation}, which is a propositional abstraction of an execution path.
Our HyperLTL semantics models a formula with a sequence of computations, making it possible to correlate multiple execution paths.

We also define a new model-checking algorithm for HyperLTL.  
Our algorithm uses a well-known LTL model-checking algorithm~\cite{Vardi:1994:infComp,Wolper00} based on {\Buchi} automata:
As input, that algorithm takes a formula $\phi$ to be verified and a system $S$ modeled as a {\Buchi} automaton $A_S$.  
The algorithm mechanically translates the formula to another {\Buchi} automaton $A_\phi$, then applies automata-theoretic constructions to $A_S$ and $A_\phi$.  
The output is either ``yes,'' the system satisfies the formula, or ``no,'' along with a counterexample path demonstrating that $\phi$ does not hold of $S$. 
In \S\ref{sec:modelchecking}, we upgrade that algorithm with a self-composition construction, so that it can verify formulas over multiple paths.
We obtain a model-checking algorithm that handles an important fragment of HyperLTL, including all of the examples in \S\ref{sec:examples}.
We implemented that algorithm in a prototype model-checker, which \S\ref{sec:modelchecking} describes.  

Hyperproperties can characterize the security policies expressible in HyperLTL.
The quantifiers appearing in a HyperLTL formula give rise to a hierarchy of hyperproperties, which we define in \S\ref{sec:hp}.  
The hierarchy contains 2-safety~\cite{TerauchiA05} and $k$-safety~\cite{ClarksonS10} hyperproperties as special cases.  
And it yields an \emph{arithmetic hierarchy} of hyperproperties that elegantly characterizes which hyperproperties can be verified by our model-checking algorithm.

This paper thus contributes to the theory of computer security by
\begin{itemize}
\item defining a new logic for expressing security polices,
\item showing that logic is expressive enough to formulate important in\-for\-ma\-tion-flow policies, 
\item giving an algorithm for model-checking a fragment of the logic, 
\item prototyping that algorithm and using it to verify security policies, and
\item characterizing the mathematical structure of security policies in terms of an arithmetic hierarchy of hyperproperties.
\end{itemize}
Though our results build upon the formal methods literature, our interest and application is entirely within the science of constructing systems that are provably secure.

We proceed as follows.
\S\ref{sec:syntax} defines the syntax of HyperLTL and provides several example formulations of information-flow policies.
\S\ref{sec:semantics} defines the semantics of HyperLTL.
\S\ref{sec:modelchecking} defines our model-checking algorithm.
\S\ref{sec:hp} discusses hyperproperties and HyperLTL.
\S\ref{sec:relatedwork} reviews related work.


\section{Syntax}
\label{sec:syntax}

HyperLTL extends propositional linear-time temporal logic (LTL)~\cite{Pnueli:1977:TLP} with explicit quantification over \emph{paths}, which are infinite sequences of execution states.
Formulas of HyperLTL are formed according to the following syntax:
\begin{eqnarray*}
\begin{aligned}
	&\phi&::=&\;\; \A\phi \mid \E\phi \mid \psi \\
	&\psi&::=&\;\;\; p\mid \neg\psi \mid \psi\vee\psi \mid \<\psi,\dots,\psi\> \mid \X\psi \mid \psi\U\psi
\end{aligned}
\end{eqnarray*}

A HyperLTL formula $\phi$ starts with a sequence of \emph{path quantifiers}.
$\A$ and $\E$ are universal and existential path quantifiers, respectively, read as ``along all paths'' and ``along some path.''
For example, $\AAEx\psi$ means that for all paths $\pi_1$ and $\pi_2$, there exists another path $\pi_3$, such that $\psi$ holds on those three paths.
(Since branching-time logics also have explicit path quantifiers, it is natural to wonder why we don't use one of them. We postpone addressing that question until~\S\ref{sec:HLTLvsOthers}.)
An atomic proposition $p$ expresses some fact about states.
The \emph{focus} connective, written $\<\psi_1,\dots,\psi_n\>$, is used to restrict attention to individual paths: $\psi_1$ must hold of the first path quantified over, $\psi_2$ of the second, and so forth.
Boolean connectives $\neg$ and $\vee$ have the usual classical meanings.
Implication, conjunction, and bi-implication are defined as syntactic sugar: $\psi_1\rightarrow\psi_2 = \neg\psi_1\vee\psi_2$, and $\psi_1\wedge\psi_2 = \neg(\neg\psi_1\vee\neg\psi_2)$, and $\psi_1 \leftrightarrow \psi_2 = \psi_1 \rightarrow \psi_2 \wedge \psi_2 \rightarrow \psi_1$.
True and false, written $\top$ and $\bot$, are defined as $p\vee\neg p$ and $\neg\top$, respectively.

Temporal connective \X{} is read as ``ne\underline{x}t.''
Formula $\X\psi$ means that $\psi$ holds on the next state of every quantified path.
Likewise, \U{} is read ``\underline{u}ntil,'' and $\psi_1\U\psi_2$ means that $\psi_2$ will eventually hold of all quantified paths, and until then $\psi_1$ holds.
The other standard temporal connectives \F{}, \G{} and \R{}, read as ``\underline{f}uture,'' ``\underline{g}lobally,'' and ``\underline{r}elease,'' are defined as syntactic sugar:
$\F\psi=\top\U\psi$, meaning in the future, $\psi$ must eventually hold; $\G\psi=\neg\F\neg\psi$, meaning $\psi$ must hold, globally; and $\psi_1\R\psi_2=\neg(\neg\psi_1\U\neg\psi_2)$, meaning $\psi_2$ must hold until released by $\psi_1$.
 
A HyperLTL formula $\phi$ is \emph{well-formed} iff
(i) $\phi$ contains at least one path quantifier, and
(ii) the length $n$ of all focus subformulas $\<\psi_1,\dots,\psi_n\>$ equals the number of path quantifiers at the beginning of $\phi$.

\subsection{Security policies in HyperLTL}
\label{sec:examples}

We now put HyperLTL into action by formulating several security policies.

\paragraph{Access control.} 
An \emph{access control} policy permits an operation $\mathit{op}$ on an object $o$ to proceed only if the subject $s$ requesting $\mathit{op}$ has the right to perform $\mathit{op}$ on $o$.
Let $\atom{permit}_{\mathit{op},o}$ be a proposition denoting that $\mathit{op}$ is permitted on $o$, and $\atom{req}_{s,\mathit{op},o}$ that $s$ has requested to perform $\mathit{op}$ on $o$, and $\atom{hasRight}_{s,\mathit{op},o}$ that $s$ has the right to perform $\mathit{op}$ on $o$.
Access control can be expressed in HyperLTL as follows:
\begin{equation}
\label{hp:ac}
\A\G (\atom{req}_{s,\mathit{op},o} \rightarrow (\atom{hasRight}_{s,\mathit{op},o} \leftrightarrow \atom{permit}_{\mathit{op},o})).
\end{equation}

\paragraph{Guaranteed service.}
If a system always eventually responds to a request for service, then it provides \emph{guaranteed service}:
\begin{equation}
\label{hp:gs}
\A\G (\atom{req} \rightarrow \F \atom{resp}).
\end{equation}
Both access control~\eqref{hp:ac} and guaranteed service~\eqref{hp:gs} are examples of trace properties expressible in LTL.
Any LTL property can be expressed in HyperLTL simply by prepending a universal path quantifier to its LTL formula.

\paragraph{Nonin(ter)ference.}
A system satisfies \emph{noninterference}~\cite{GoguenMeseguer:1982:NI}  when the outputs observed by low-security users are the same as they would be in the absence of inputs submitted by high-security users.
Noninterference thus requires a system to be closed under \emph{purging} of high-security inputs.
The original formulation~\cite{GoguenMeseguer:1982:NI} of noninterference uses an \emph{event-based} system model, in which an execution path is a sequence of individual events (e.g., commands), and purging removes high-security events from the sequence.
An alternative formulation~\cite{McLean:1994:GeneralTheory} called \emph{noninference} uses a \emph{state-based} system model, in which an execution path is a sequences of states (e.g., values of variables), and purging assigns an ``empty'' value, denoted $\lambda$, to the high-security component of the state.
Noninterference and noninference are both intended to be used with deterministic systems.

Here, we pursue the state-based model, because it blends well with temporal logic, which is also based on states.  
We note that for any event-based system, there is a state-based system equivalent to it~\cite{Millen94}, though an infinite number of states might be required.

Inputs, outputs, and users are classified into \emph{security levels} in the following examples.
For simplicity, we consider only two levels, \emph{high} and \emph{low}. 
We assume that each state contains input and output variables of each security level.

Let $\atom{high}$ hold in a state when its high inputs and outputs are not $\lambda$, and let $\atom{low-equiv}$ hold on a pair of states whenever those states have the same low inputs and outputs. 
Using those propositions, noninference can be expressed as follows:
\begin{equation}
\label{hp:ni}
\AEx\G(\<\top,\neg\atom{high}\>\wedge\atom{low-equiv}).
\end{equation}
The formula starts with $\AEx$, which means ``for all paths, there exists another path.''
Low equivalence of those paths is formulated as $\G\atom{low-equiv}$, which means that at each time step, the current states in the two paths are low equivalent.
Subformula $\<\top,\neg\atom{high}\>$ requires all states of the second path to have empty high inputs and outputs.
The second path is therefore the first path, but with its high inputs and outputs purged.

\paragraph{Nondeterminism and noninterference.}
Goguen and Meseguer's definition of noninterference~\cite{GoguenMeseguer:1982:NI} requires systems to be deterministic.
Nondeterminism is useful for specification of systems, however, so many variants of noninterference have been developed for nondeterministic systems.

A (nondeterministic) system satisfies \emph{observational determinism}~\cite{ZdancewicMyers:2003:OD} if every pair of executions with the same initial low observation remain indistinguishable by low users.
That is, the system appears to be deterministic to low users.
Systems that satisfy observational determinism are immune to \emph{refinement attacks}~\cite{ZdancewicMyers:2003:OD}, because observational determinism is preserved under refinement.
Observational determinism can be expressed as follows: 
\begin{equation}
\label{hp:od}
\AAx\atom{low-equiv}\rightarrow\G\atom{low-equiv}.
\end{equation}

There are many definitions of noninterference that do permit low-ob\-serv\-able nondeterminism.
\emph{Generalized noninterference} (GNI)~\cite{McCullough:1987:GNI}, for example, stipulates that the low-security outputs may not be altered by the injection of high-security inputs.
Like noninterference, GNI was original formulated for event-based systems, but it can also be formulated for state-based systems~\cite{McLean:1994:GeneralTheory}.  
GNI can be expressed as follows:
\begin{equation}
\label{hp:gni}
\AAEx\G(\atom{high-in-equiv}_{1,3}\wedge\atom{low-equiv}_{2,3}).
\end{equation}
Proposition $\atom{high-in-equiv}_{1,3}$ holds when the current states of the first and third paths have the same high inputs, and $\atom{low-equiv}_{2,3}$ holds when the current states in the second and third paths are low equivalent.
The third path is therefore an \emph{interleaving} of the high inputs of the first path and the low inputs and outputs of the second path.
Other security policies based on interleavings, such as \emph{restrictiveness}~\cite{McCullough:1990:Hookup} and \emph{separability}~\cite{McLean:1994:GeneralTheory}, can similarly be expressed in HyperLTL.


\subsection{Comparison with other temporal logics}
\label{sec:HLTLvsOthers}

Why did we invent a new temporal logic instead of using an existing, well-studied logic?
In short, because we don't know of an existing temporal logic that can directly express all the policies in \S\ref{sec:examples}:
\begin{itemize}
\item \textbf{Linear time.} 
LTL formulas express properties of individual execution paths.
But all of the noninterference properties of \S\ref{sec:examples} are properties of sets of execution paths~\cite{McLean:1994:GeneralTheory,ClarksonS10}, hence cannot be formulated in LTL.
Explicit path quantification does enable their formulation in HyperLTL.
\item \textbf{Branching time.} 
CTL~\cite{EmersonClarke:1982:CTL} and CTL$^*$~\cite{EmersonH86} have explicit path quantifiers.
But their quantifiers don't enable expression of relationships between paths, because only one path is ``in scope'' at a given place in a formula.
(See appendix~\ref{sec:whynotctl} for an example.)
So they can't directly express policies such as observational determinism~\eqref{hp:od} and GNI~\eqref{hp:gni}.
HyperLTL does allow many paths to be in scope, as well as propositions over all those paths.
\par
By using the self-composition construction, it is possible to express relational noninterference in CTL~\cite{BartheDR04} and observational determinism in CTL$^*$~\cite{HuismanWS06}.
Those approaches resemble HyperLTL, but HyperLTL formulas express policies directly over the original system, rather than over a self-composed system.
\par
Furthermore, the self-composition approach does not seem capable of expressing policies, such as noninference~\eqref{hp:ni} and generalized noninterference~\eqref{hp:gni}, that have both universal and existential path quantifiers over infinite paths.
(A recent upgrade of self-composition, \emph{asymmetric product programs}~\cite{BartheCK13}, does enable verification of refinement properties involving both kinds of quantifiers.
It might be possible to express policies like noninference with that upgrade.)
Nonetheless, it is straightforward to express such policies in HyperLTL.
\item \textbf{Modal $\mu$-calculus.}
Modal $\mu$-calculus\cite{Kozen:1982:mu-calc} generalizes CTL$^*$.
But as expressive as modal $\mu$-calculus is, it remains insufficient~\cite{AlurCZ06} to express all \emph{opacity} policies~\cite{BryansKMR05}, which prohibit observers from discerning the truth of a predicate.
(Alur et al.~\cite{AlurCZ06} actually write ``secrecy'' rather than ``opacity.'')
Simplifying definitions slightly, a trace property $P$ is \emph{opaque} iff for all paths $\pi$ of a system, there exists another path $\pi'$ of that system, such that $\pi$ and $\pi'$ are low-equivalent, and exactly one of $\pi$ and $\pi'$ satisfies $P$.
HyperLTL is able to express all opacity policies over linear-time properties:
given LTL formula $\phi_P$ that expresses a linear-time trace property $P$, HyperLTL formula $$\AEx ((\G \atom{low-equiv}) \wedge (\<\phi_P, \neg \phi_P\> \vee \<\neg\phi_P, \phi_P\>))$$ stipulates that $P$ is opaque.
Noninference~\eqref{hp:ni}, for example, is a linear-time opacity policy~\cite{PeacockR06}.
\end{itemize}


\section{Semantics}
\label{sec:semantics}

HyperLTL formulas are interpreted with respect to \emph{computations}.
A computation abstracts away from the states in a path, representing each state by the propositions that hold of that state.
Let $\Atoms$ denote the set of atomic propositions. 
Formally, a computation $\gamma$ is an infinite sequence over $\powerset(\Atoms)$, where $\powerset$ denotes the powerset operator.
Define $\gamma[i]$ to be element $i$ of computation $\gamma$.  
Hence, $\gamma[i]$ is a set of propositions.
And define $\gamma[i..]$ to be the suffix of $\gamma$ starting with element $i$---that is, the sequence $\gamma[i]\gamma[i+1]\ldots$
We index sequences starting at 1, so $\gamma[1..] = \gamma$.

A computation represents a single path, but HyperLTL formulas may quantify over multiple paths.
To represent that, let $\Gamma$ denote a finite tuple $(\gamma_1,\dots,\gamma_k)$ of computations.
Define $|\Gamma|$ to be the length $k$ of $\Gamma$,
and define projection $\proj_i(\Gamma)$ to be element $\gamma_i$.
Given a tuple $\Gamma$ define $\Gamma\cdot\gamma$ to be the concatenation of element $\gamma$ to the end of tuple $\Gamma$, yielding tuple $(\gamma_1,\dots,\gamma_k,\gamma)$.
Extend that notation to concatenation of tuples by defining  $\Gamma\cdot\Gamma'$ to be the tuple containing all the elements of $\Gamma$ followed by all the elements of $\Gamma'$.
Extend notations $\gamma[i]$ and $\gamma[i..]$ to apply to computation tuples by defining $\Gamma[i] = (\gamma_1[i],\dots,\gamma_k[i])$---that is, the tuple containing element $i$ from each computation in $\Gamma$---and $\Gamma[i..] = (\gamma_1[i..],\dots,\gamma_k[i..])$.

HyperLTL formulas may involve propositions over multiple states.  
For example, $\atom{low-equiv}$ in the definition of noninference~\eqref{hp:ni} holds when two states have the same low inputs and outputs. 
We therefore need a means to determine what \emph{compound} propositions hold of a tuple of states, given what atomic propositions hold of the individual states.
To do that, we introduce \emph{bonding} functions that describe how to produce compound propositions out of tuples of atomic propositions.
Let $\Compounds$ denote the set of compound propositions, and assume that $\Atoms \subseteq \Compounds$. 
Let $B$ be a family $\setdef{B_i}{i\in\mathbb{N}}$ of functions, such that each $B_i$ is a function from $\mathcal{P}(\Atoms)^i$ to $\mathcal{P}(\Compounds)$.
Notation $X^n$ is the $n$-ary cartesian power of set $X$.
We require $B_1$ to be the identity function, so that length-1 tuples are not changed by bonding.
As an example, consider a bonding function $B_2$ that describes when two states are low-equivalent.
Given a set $\setdef{\atom{low}_i}{1 \leq i \leq n}$ of atoms, describing $n$ different low states, we could define $B_2$ such that $B_2(\{\atom{low}_i\},\{\atom{low}_i\}) = \{\atom{low-equiv}\}$, and $B_2(\{\atom{low}_i\},\{\atom{low}_j\}) = \emptyset$ if $i \neq j$.
Given a tuple, it is always clear from the length of the tuple which function $B_i$ should be applied to it, so henceforth we omit the subscript.

The validity judgment for HyperLTL formulas is written $\Gamma \models \phi$.  
Formula $\phi$ must be well-formed.
The judgment implicitly uses a \emph{model} $M$, which is a set of computations, and a family $B$ of bonding functions.
We omit notating $M$ and $B$ as part of the judgment, because they do not vary during the interpretation of a formula.  
Validity is defined as follows:
\begin{enumerate}
\item\label{validity:A} $\Gamma \models \A\psi$ iff for all $\gamma\in M : \Gamma\cdot\gamma\models\psi$ 
\item\label{validity:E} $\Gamma \models \E\psi$ iff there exists $\gamma\in M : \Gamma\cdot\gamma\models\psi$
\item\label{validity:p} $\Gamma \models p$ iff $p \in B(\Gamma[1])$
\item\label{validity:not} $\Gamma \models \neg \psi$ iff $\Gamma\not\models \psi$ 
\item\label{validity:or} $\Gamma \models \psi_1\vee\psi_2$ iff $\Gamma \models\psi_1$ or $\Gamma\models \psi_2$
\item\label{validity:focus} $\Gamma \models \<\psi_1,\dots,\psi_n\>$ iff for all $i$ : if $1 \leq i \leq n$ then $\proj_i(\Gamma)\models\psi_i$ 
\item\label{validity:X} $\Gamma \models \X\psi$ iff $\Gamma[2..]\models\psi$
\item\label{validity:U} $\Gamma \models \psi_1\U\psi_2$ iff there exists $k : k \geq 1$ and $\Gamma[k..]\models\psi_2$ and for all $j$ : if $1 \leq j < k$ then $\Gamma[j..] \models\psi_1$
\end{enumerate}
Clauses \ref{validity:A} and \ref{validity:E} quantify over a computation $\gamma$ from $M$, and they concatenate $\gamma$ to $\Gamma$ to evaluate subformula $\psi$.
Clause \ref{validity:p} means satisfaction of atomic propositions is determined by the first element of each computation in $\Gamma$.
Clauses \ref{validity:not} and \ref{validity:or} are standard.
In clause \ref{validity:focus}, elements of a focus formula are independently evaluated over their corresponding individual computations.
Clauses \ref{validity:X} and \ref{validity:U} are the standard LTL definitions of $\X$ and $\U$, upgraded to work over a sequence of computations.


\section{Model Checking}
\label{sec:modelchecking}


Model-checking is possible at least for fragments of HyperLTL.
For example, HyperLTL contains LTL as a fragment, and LTL enjoys a decidable model-checking algorithm.
Here's a much larger fragment of HyperLTL that can be model checked:
\begin{itemize}
\item The series of quantifiers at the beginning of a formula may involve only a single alternation of quantifiers.
For example, $\E\psi$ and $\AAEx\psi$ are allowed, but $\AEAx\psi$ is not. 
\item In focus formulas $\<\psi_1,\dots,\psi_n\>$, the subformulas $\psi_i$ may not use temporal connectives $\X$ and $\U$.  Hence all the $\psi_i$ must be propositional formulas.
\end{itemize}
We name this fragment \emph{\Hp}, because every formula in it may begin with at most two kinds of quantifiers---a sequence of \A's followed by a sequence of \E's, or vice-versa.
{\Hp} is an important fragment, because it is expressive enough for all the security policies formulated in \S\ref{sec:examples}.

We now give a model-checking algorithm for {\Hp}.  
Our algorithm adapts previously known algorithms for LTL model-checking~\cite{GerthVardi:1995:OntheFlyLTL,Vardi:2007:automata-theoretic-rev,Vardi:1996:LTL,Gastin:2001:FastLTLtoBuchi}.
Those LTL algorithms determine whether a set $M$ of computations satisfies an LTL formula $\phi$, as follows:
\begin{enumerate}
\item Represent $M$ as a {\Buchi} automaton~\cite{Buchi62}, $A_M$.
Its language is $M$.
\item Construct {\Buchi} automaton $A_{\neg\phi}$, whose language is the set of all computations that don't satisfy $\phi$.
\item Intersect $A_M$ and $A_{\neg\phi}$, yielding automaton $A_M \cap A_{\neg\phi}$.  Its language contains all computations of $M$ that don't satisfy $\phi$.
\item Check whether the language of $A_M \cap A_{\neg\phi}$ is empty. 
If so, all computations of $M$ satisfy $\phi$, hence $M$ satisfies $\phi$.
If not, then any element of the language is a counterexample showing that $M$ doesn't satisfy $\phi$.
\end{enumerate}

Our algorithm for model-checking {\Hp} adapts that LTL algorithm.
Without loss of generality, assume that the {\Hp} formula to be verified has the form $\Ak \Ej \psi$, where $\Ak$ and $\Ej$ denote sequences of universal and existential path quantifiers of lengths $k$ and $j$.
(Formulas of the form $\Ek\Aj\psi$ can be verified by rewriting them as $\Ak\Ej\neg\psi$.)
Let $n$ equal $k+j$.
Semantically, a model of $\psi$ must be an $n$-tuple of computations.
Let $\zip$ denote the usual function that maps an $n$-tuple of sequences to a single sequence of $n$-tuples---for example, $\zip([1,2,3],[4,5,6]) = [(1,4), (2,5), (3,6)]$---and let $\unzip$ denote its inverse.
To determine whether a system $M$ satisfies {\Hp} formula $\Ak \Ej \psi$, our algorithm follows the same basic steps as the LTL algorithm:
\begin{enumerate}
\item Represent $M$ as a {\Buchi} automaton, $A_M$. 
Construct the $n$-fold product of $A_M$ with itself---that is, $A_M \times A_M \times \cdots \times A_M$, where ``$A_M$'' occurs $n$ times.
This construction is straightforward and formalized in appendix~\ref{sec:constructions}. 
Denote the resulting automaton as $A^n_M$.
If $\gamma_1,\ldots\gamma_n$ are all computations of $M$, then $\zip(\gamma_1,\ldots\gamma_n)$ is a word in the language of $A^n_M$.%

\item Construct {\Buchi} automaton ${A_{\psi}}$.
Its language is the set of all words $w$ such that $\unzip(w)=\Gamma$ and $\Gamma \models \psi$---that is, the tuples $\Gamma$ of computations that satisfy $\psi$.
This construction, formalized in appendix~\ref{sec:constructions}, is a generalization of the corresponding LTL construction.

\item Intersect $A^n_M$ and ${A_{\psi}}$, yielding automaton ${A^n_M\cap A_{\psi}}$.  
Its language is essentially the tuples of computations of $M$ that satisfy $\psi$.  
This construction is standard~\cite{Clarke:1999:model-checking}. 

\item Check whether $\mathcal{L}({((A^n_M\cap A_{\psi})|_k)^C\cap A^k_M})$ is empty, where
(i) $A^C$ denotes the \emph{complement} of an automaton $A$,
(complement constructions are well-known---e.g.,~\cite{Vardi:1996:LTL}---so we do not formalize one here), 
and (ii) $A|_k$ denotes the same automaton as $A$, but with every transition label (which is an $n$-tuple of propositions) \emph{projected} to only its first $k$ elements.
That is, if $\mathcal{L}(A)$ contains words of the form $\zip(\gamma_1,\ldots\gamma_n)$, then $\mathcal{L}(A|_k)$ contains  words of the form $\zip(\gamma_1,\ldots\gamma_k)$.
Projection erases the final $j$ computations from each letter of a word, leaving only the initial $k$ computations.
Thus a word is in the projected language iff there exists some extension of the word in the original language. 

If $\mathcal{L}({((A^n_M\cap A_{\psi})|_k)^C\cap A^k_M})$ is empty, then $M$ satisfies $\Ak \Ej \psi$.
If not, then any element of the language is a counterexample showing that $M$ doesn't satisfy $\Ak \Ej \psi$.

\end{enumerate}

The final step of the above algorithm is a significant departure from the LTL algorithm.  
Intuitively, it works because projection introduces an existential quantifier, thus enabling verification of formulas with a quantifier alternation.
The following theorem states the correctness of our algorithm:
\begin{theorem}\label{thm:modelcheckingcorrect}
Let $\phi$ be {\Hp} formula $\Ak\Ej\psi$, and let $n = k+j$.  
Let $M$ be a set of computations.
Then $\phi$ holds of $M$ iff $\mathcal{L}({((A^n_M\cap A_{\psi})|_k)^C\cap A^k_M})$ is empty.
\end{theorem}
\begin{proof}
($\Rightarrow$, by contrapositive)
We seek a countermodel showing that $\Ak\Ej\psi$ doesn't hold of $M$.
For that countermodel to exist, 
\begin{equation}\label{eq:mc1}
\begin{split}
&\text{there must exist a $k$-tuple $\Gamma_k$} : \\ &\text{for all $j$-tuples $\Gamma_j$ : if $\mathit{set}(\Gamma_k \cdot \Gamma_j) \subseteq M$ then $\Gamma_k \cdot \Gamma_j \models \neg\psi$,}
\end{split}
\end{equation}
where $\mathit{set}(\Gamma)$ denotes the set containing the same elements as tuple $\Gamma$.
To find that countermodel $\Gamma_k$, consider $\mathcal{L}(A^n_M\cap A_{\psi})$.
If that language is empty, then 
\begin{equation}\label{eq:mc2}
\begin{split}
&\text{for all $k$-tuples $\Gamma_k$ and} \\ &\text{for all $j$-tuples $\Gamma_j$ : }\text{if $\mathit{set}(\Gamma_k \cdot \Gamma_j) \subseteq M$ then $\Gamma_k \cdot \Gamma_j \models \neg\psi$.}
\end{split}
\end{equation}
That's almost what we want, except that $\Gamma_k$ is universally quantified in~\eqref{eq:mc2} rather than existentially quantified as in~\eqref{eq:mc1}.
So we introduce projection and complementation to relax the universal quantification to existential.
First, note that language $\mathcal{L}((A^n_M\cap A_{\psi})|_k)$ contains all $\zip(\Gamma_k)$ for which there exists a $\Gamma_j$ such that $\mathit{set}(\Gamma_k \cdot \Gamma_j) \subseteq M$ and $\Gamma_k \cdot \Gamma_j \models \psi$.
So if there exists a $\Gamma_k^*$ such that $\zip(\Gamma_k^*) \not\in \mathcal{L}((A^n_M\cap A_{\psi})|_k)$, then for all  $\Gamma_j$, if $\mathit{set}(\Gamma_k \cdot \Gamma_j) \subseteq M$ then $\Gamma_k \cdot \Gamma_j \models \neg\psi$.
That $\Gamma_k^*$ would be exactly the countermodel we seek according to~\eqref{eq:mc1}.
To find such a $\Gamma_k^*$, it suffices to determine whether $\mathcal{L}((A^n_M\cap A_{\psi})|_k) \subset \mathcal{L}(A^k_M)$, because any element that strictly separates those sets would satisfy the requirements to be a $\Gamma_k^*$.
By simple set theory, $X \subset Y$ iff $X^C \cap Y$ is not empty.
Therefore, if $\mathcal{L}({((A^n_M\cap A_{\psi})|_k)^C\cap A^k_M})$ is not empty, then a countermodel $\Gamma_k^*$ exists.

($\Leftarrow$) The same argument suffices: if $\mathcal{L}({((A^n_M\cap A_{\psi})|_k)^C\cap A^k_M})$ is empty, then no countermodel can exist.
\end{proof}

We are currently investigating the complexity of this model-checking algorithm.

\paragraph*{Formulas without quantifier alternation.} 
Define HyperLTL$_1$ to be the fragment of {\Hp} that contains formulas with no alternation of quantifiers.
HyperLTL$_1$ can be verified more efficiently than {\Hp}.
Suppose $\phi$ is HyperLTL$_1$ formula $\A^n\psi$. 
Then it suffices to check whether $A^n_M \cap A_{\neg\psi}$ is non-empty.
This is essentially the self-composition construction, as used in previous work~\cite{BartheDR04,TerauchiA05,ClarksonS10}.

\paragraph*{Prototype.}
We implemented a prototype for the model-checking algorithm in OCaml. The prototype accepts an input file for the state transition system description, and a \Hp{} formula. For the prototype, the description language of the state transition system requires explicit definition of the states, single-state and multistate labels. For automata complementation, the prototype uses GOAL~\cite{GOAL:2007}, an interactive tool for manipulating \Buchi{} automata. In the case that a \Hp{} property doesn't hold, a witness will be produced.


\section{Hyperproperties}
\label{sec:hp}

The mathematical structure of the class of security policies expressible in HyperLTL can be precisely characterized by hyperproperties. 
We begin by summarizing the theory of hyperproperties.

\begin{definition}[Hyperproperties~\cite{ClarksonS10}]
A \emph{trace} is a finite or infinite sequence of states.  
(The terms ``infinite trace'' and ``path'' are therefore synonymous.)
Define $\trfin$ to be the set of finite traces and $\trinf$ to be the set of infinite traces.
A \emph{trace property} is a set of infinite traces. 
A set $T$ of traces satisfies a trace property $P$ iff $T \subseteq P$.
A \emph{hyperproperty} is a set of sets of infinite traces, or equivalently a set of trace properties. 
The interpretation of a hyperproperty as a security policy is that the hyperproperty is the set of systems allowed by that policy.
Each trace property in a hyperproperty is an allowed system, specifying exactly which executions must be possible for that system. 
Thus a set $T$ of traces satisfies hyperproperty $\hp{H}$ iff $T$ is in $\hp{H}$.
Given a trace property $P$, the powerset of $P$ is the unique hyperproperty that expresses the same policy as $P$.  Denote that hyperproperty as $\lift{P}$.
\end{definition}

\subsection{k-hyperproperties}
\label{sec:khp}

A system satisfies a trace property if every trace of the system satisfies the property.  
To determine whether a trace satisfies the property, the trace can be considered in isolation, without regard for any other traces that might or might not belong to the system.
Similarly, a system satisfies observational determinism if every pair of its traces---where every pair can be considered in isolation---satisfies HyperLTL formula~\eqref{hp:od}.
These examples suggest a new class of hyperproperties based on the idea of satisfaction determined by bounded sets of traces.

Let a \emph{$k$-hyperproperty} be a hyperproperty that is definable by a $k$-ary relation on traces as follows.
Intuitively, one needs to consider at most $k$ traces at a time to decide whether a system satisfies a $k$-hyperproperty.
Formally, a hyperproperty \hp{H} is a $k$-hyperproperty iff 
\begin{multline*}
\existsqer{R \subseteq \trinf^k}{\forallqer{S \in \hp{H}}{\\ \forallq{\vec{t} \in \trinf^k}{
\mathit{set}(\vec{t}\;) \subseteq S}{\vec{t} \in R}}}, 
\end{multline*}
where $\trinf^k$ denotes the $k$-fold Cartesian product of $\trinf$ (i.e., the set of all $k$-tuples of infinite traces), $\vec{t}$ denotes a $k$-tuple $(t_1,\ldots,t_k)$ of traces, and $\mathit{set}(\vec{t}\;)$ denotes $\setdef{t_i}{1 \leq i \leq k}$.
For a system $S$ to be a member of $\hp{H}$, all $k$-tuples of traces from $S$ must satisfy $R$, in which case relation $R$ \emph{defines} \hp{H}.

Trace properties are 1-hyperproperties:  to decide whether a system $S$ satisfies a 1-hyperproperty, it suffices to consider each trace of $S$ in isolation. 
For a trace property $P$, the relation that defines hyperproperty $[P]$ is $P$ itself, because 
\begin{equation*}
{\forallqer{S \in [P]}{\forallq{\vec{t} \in \trinf^1}{\mathit{set}(\vec{t}\;) \subseteq S}{\vec{t} \in P}}}.
\end{equation*}
 
The $k$-hyperproperties form a hierarchy in which each level requires consideration of one more trace than the previous level.
Formally, any $k$-hyperproperty defined by $R$ is also a $(k+1)$-hyperproperty defined by relation $\setdef{\vec{t}\cdot u}{\vec{t} \in R \andsp u \in \trinf}$, where $\cdot$ denotes appending an element to a tuple---that is, $(t_1, \ldots, t_k)\cdot u = (t_1, \ldots, t_k, u)$.
So all 1-hyperproperties are also 2-hyperproperties, etc.  

Observational determinism is a 2-hyperproperty, because it suffices to consider pairs of traces to decide whether a system satisfies it. 
The HyperLTL formula~\eqref{hp:od} that characterizes it makes this apparent:
\begin{itemize}
\item The two quantifiers at the beginning of the formula, $\AAx$, show that the policy is defined in terms of pairs of traces.
\item The subformula following the quantifiers, $\atom{low-equiv} \Rightarrow \G \atom{low-equiv}$  gives the relation that defines the policy as a 2-hyperproperty.  
That relation is the set of all pairs $(t_1,t_2)$ of traces such that $$\comp(t_1),\comp(t_2) \models \atom{low-equiv} \Rightarrow \G \atom{low-equiv}.$$
\end{itemize}

Noninference, however, is not a 2-hyperproperty.  
Though it can be defined as a relation on pairs of traces, one of those traces is existentially quantified;
$k$-hyperproperties allow only universal quantification.
That suggests the following generalization of $k$-hyperproperties.

\subsection{Q-hyperproperties}

Let $Q$ be a finite sequence of universal and existential quantifiers---for example, $\forall\exists$.
Define hyperproperty \hp{H} to be an \emph{$Q$-hyperproperty} iff $|Q| = k$ and 
\begin{equation*}
\existsqer{R \subseteq \trinf^k}{\forallqer{S \in \hp{H}}
{\quant{Q}{\vec{t} \in \trinf^k}
{\vec{t} \subseteq S}{\vec{t} \in R}{:}{:}}}.
\end{equation*}
Notation $Q \, \vec{t} \in \trinf^k$ is an abbreviation for $k$ nested quantifications:
\begin{equation*}
Q_1 \, t_1 \in \trinf \;\! : \;\! Q_2 \, t_2 \in \trinf \;\! : \;\! \ldots \;\! : \;\! Q_k \, t_k \in \trinf,
\end{equation*}
where $Q_i$ denotes quantifier $i$ from sequence $Q$.
For example, $\forall\exists \, \vec{t} \in \trinf^k$ abbreviates $\forall \, t_1 \in \trinf \;\! : \;\! \exists \, t_2 \in \trinf$.

Noninference is a $\forall\exists$-hyperproperty.
The HyperLTL formula~\eqref{hp:ni} that characterizes it makes this apparent.
The two quantifiers show that the policy is defined in terms of pairs of traces. 
Its defining relation is the set of all pairs $(t_1,t_2)$ of traces such that $$\comp(t_1),\comp(t_2) \models \G(\<\top,\neg\atom{high-in}\>\wedge\atom{low-equiv}).$$
Likewise, separability and generalized noninterference are both $\forall\forall\exists$-hyper\-properties, and restrictiveness is the intersection of two $\forall\forall\exists$-hyperproperties. 

The $Q$-hyperproperties strictly generalize the $k$-hyper\-prop\-er\-ties, because (i) for all $k$, a $k$-hyperproperty is a $\forall^k$-hyper\-prop\-er\-ty, where $\forall^k$ denotes a sequence of $k$ universal quantifiers, and because (ii) no $Q$-hyperproperty, such that $Q$ contains $\exists$, is a $k$-hyperproperty.
As do the $k$-hyperproperties, the $Q$-hyperproperties form a hierarchy:  any $Q$-hyper\-prop\-er\-ty is also a $Q'$-hyperproperty if sequence $Q$ is a prefix of sequence $Q'$.

A $Q$-hyperproperty is \emph{linear-time} if its defining relation $R$ is linear-time, meaning that it can be expressed with the linear-time temporal connectives, $\X$ and $\U$. (Or, equivalently~\cite{GabbayPSS80}, that it can be expressed in S1S, the monadic second-order theory of one successor.)
\begin{proposition}
$\hp{H}$ is a linear-time $Q$-hyperproperty iff there exists a HyperLTL formula $\phi$, such that $S \in \hp{H}$ iff $S \models \phi$.
\end{proposition}
\begin{proof}
The relation $R$ that defines $\hp{H}$ is equivalent to  formula $\phi$.
\end{proof}
\noindent HyperLTL therefore expresses exactly the linear-time $Q$-hyperproperties, just as LTL expresses exactly the linear-time trace properties, which are themselves $\forall$-hyperproperties.

\subsection{Safety}

\emph{Safety}~\cite{AlpernS85} proscribes ``bad things.'' 
A bad thing is \emph{finitely observable}, meaning its occurrence can be detected in finite time, and \emph{irremediable}, so its occurrence can never be remediated by future events.

\begin{definition}[Hypersafety~\cite{ClarksonS10}]
A hyperproperty $\hp{S}$ is a \emph{safety hyperproperty} (is
\emph{hypersafety}) iff
\begin{multline*}
\forallq{T \subseteq \trinf}{T \notin \hp{S}}{\existsq{B \subseteq \trfin }{|B| \in \mathbb{N} \andsp B \leq T \\ }{\forallq{U \subseteq \trinf}{B \leq U}{U \notin
\hp{S}}}}.
\end{multline*}
For a system $T$ that doesn't satisfy a safety hyperproperty, the bad thing is a finite set $B$ of finite traces.  
$B$ cannot be a prefix of any system $U$ satisfying the hyper safety property.
A finite trace $t$ is a \emph{prefix} of a (finite or infinite) trace $t'$, denoted $t \prefix t'$, iff $t' = tt''$ for some $t'' \in \tr$.
And a finite set $T$ of finite traces is a prefix of a (finite or infinite) set $T'$ of (finite or infinite) traces, denoted $T \leq T'$, iff $\forallqer{t \in T}{\existsqer{t' \in T'}{t \leq t'}}$.
A $k$-safety hyperproperty is a safety hyperproperty in which the bad thing never involves more than $k$ traces.
A hyperproperty $\hp{S}$ is a \emph{$k$-safety hyperproperty} (is \emph{$k$-hypersafety})~\cite{ClarksonS10} iff
\begin{multline*}
\forallq{T \subseteq \trinf}{T \notin \hp{S}}{\existsq{B \subseteq \trfin }{|B| \leq k \andsp B \leq T\\ }{\forallq{U \subseteq \trinf}{B \leq U}{U \notin \hp{S}}}}.
\end{multline*}
\noindent This is just the definition of hypersafety but with the cardinality of $B$ bounded by $k$.
\end{definition}

Define a relation $R$ to be a \emph{$k$-safety relation} iff
\begin{multline*}
\forallq{\vec{t} \in \trinf^k}{\vec{t} \not\in R}
 {\existsq{\vec{b} \in \trfin^{\leq k}}{\vec{b} \leq \vec{t}}
   {\\ \forallq{\vec{u} \in \trinf^k}{\vec{b} \leq \vec{u}}{\vec{u} \not\in R}}}.
\end{multline*}
Prefix on tuples of traces is the pointwise application of prefix on traces:  $\vec{t} \leq \vec{u}$ iff, for all $i$, it holds that $t_i \leq u_i$.
Set $\trfin^{\leq k}$ is all $n$-tuples of traces where $n \leq k$.

Observational determinism is a 2-safety hyperproperty~\cite{ClarksonS10}, as well as a 2-hyperproperty definable by a 2-safety relation.  Moreover, the $k$-safety hyperproperties are all $k$-hyper\-prop\-er\-ties:
\begin{proposition}
A $k$-hyperproperty $\hp{H}$ is definable by a $k$-safety relation iff $\hp{H}$ is a $k$-safety hyperproperty.
\end{proposition}
\begin{proof}
($\Rightarrow$) The bad thing for a system that doesn't satisfy $\hp{H}$ is tuple $\vec{b}$.
($\Leftarrow$) The relation is the set of all $k$-tuples of traces that do not contain a bad thing as a prefix.
\end{proof}
The $k$-safety hyperproperties are known~\cite{ClarksonS10} to have a relatively complete verification methodology based on self-composition.
Our model-checking algorithm in \S\ref{sec:modelchecking} increases the class of hyperproperties that can be verified from $k$-safety to $k$-hyperproperties and a fragment of $Q$-hyper\-prop\-er\-ties.

\subsection{Arithmetic hierarchy}

The $Q$-hyperproperties categorize by quantifier structure.
The \emph{arithmetic hierarchy}, first studied by Kleene~\cite{Kleene43}, similarly categorizes computable relations.
Rogers~\cite{Rogers87} gives the following characterization of the arithmetic hierarchy:  
\begin{definition}[Arithmetic hierarchy~\cite{Rogers87}]
An $n$-ary relation $S$ is in the arithmetic hierarchy iff $S$ is decidable or there exists a decidable $k$-ary relation $R$ such that 
\begin{multline*}
S = \setdef{(x_1, x_2, \ldots, x_n)}{Q_1 y_1 Q_2 y_2 \ldots Q_k y_k : \\ R(x_1, \ldots, x_n, y_1, \ldots, y_k)},
\end{multline*}
where, for all $1 \leq i \leq k$, quantifier $Q_i$ is either $\forall$ or $\exists$.
The sequence of quantifiers $Q_i$ is the \emph{quantifier prefix}.  
When such a prefix does exist for $S$, then $S$ is \emph{expressible by $Q_i$}.
The \emph{number of alternations} in a prefix the number of pairs of adjacent but unlike quantifiers.
For example, in the prefix $\forall\forall\exists\forall$, there are two alternations.
A \emph{$\Sigma_n$-prefix}, where $n > 0$, is a prefix that begins with $\exists$ and has $n-1$ alternations.
A \emph{$\Sigma_0$-prefix} is a prefix that is empty.
Likewise, a \emph{$\Pi_n$-prefix}, where $n > 0$, is a prefix that begins with $\forall$ and has $n-1$ alternations.
A \emph{$\Pi_0$-prefix} is a prefix that is empty, so $\Pi_0$-prefixes are the same as $\Sigma_0$-prefixes.
The arithmetic hierarchy comprises the following classes:
\begin{itemize}
\item $\Sigma_n$ is the class of all relations expressible by $\Sigma_n$-prefixes.
\item $\Pi_n$ is the class of all relations expressible by $\Pi_n$-prefixes.
\item (Another class, $\Delta_n = \Sigma_n \cap \Pi_n$, does not concern us here.)
\end{itemize}
A relation expressible by $\exists\forall$ is, for example, in $\Sigma_2$, and a relation expressible by $\forall$ is in $\Pi_1$, but a relation expressible by $\forall\forall$ is also in $\Pi_1$.
These classes form a hierarchy, because they grow strictly larger as $n$ increases: $\Sigma_n \subset \Sigma_{n+1}$ and $\Pi_n \subset \Pi_{n+1}$.
\end{definition}

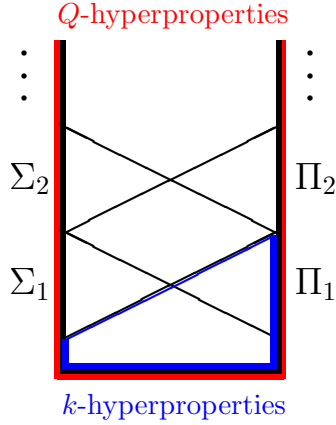
\begin{figure}
\setlength{\unitlength}{1mm}
\begin{center}
\begin{picture}(40,60)

	{
	\color{red}
	\linethickness{0.7mm}
	\put(5,5){\line(1,0){30}}
	\put(5,4.7){\line(0,1){45}}
	\put(35,4.7){\line(0,1){45}}
	
	\put(5,52){$Q$-hyperproperties}
	}
	
	{
	\linethickness{0.7mm}
	\put(5.7,5.7){\line(1,0){28.7}}
	\put(5.7,5.4){\line(0,1){44.2}}
	\put(34.4,5.4){\line(0,1){44.2}}

	\thicklines
	\put(5.7,10){\line(2,1){29}}
	\put(5.7,24){\line(2,1){29}}
	\put(34.4,10){\line(-2,1){29}}
	\put(34.4,24){\line(-2,1){29}}
	}

	{
	\color{blue}
	\linethickness{0.7mm}
	\put(3.5,6.4){\line(1,0){28}}
	\put(3.5,6){\line(0,1){3.9}}
	\put(31.1,6){\line(0,1){17.7}}

	\thicklines
	\put(3.5,9.8){\line(2,1){27.5}}
	
	\put(3,0){$k$-hyperproperties}
	}
	
	\put(-4,16){\Large $\Sigma_1$}
	\put(-4,30){\Large $\Sigma_2$}
	\put(-2,42){\circle*{0.8}}
	\put(-2,45){\circle*{0.8}}
	\put(-2,48){\circle*{0.8}}

	\put(34,16){\Large $\Pi_1$}
	\put(34,30){\Large $\Pi_2$}
	\put(36,42){\circle*{0.8}}
	\put(36,45){\circle*{0.8}}
	\put(36,48){\circle*{0.8}}

\end{picture}
\end{center}
\caption{Arithmetic hierarchy of hyperproperties\label{fig:hierarchy}}
\end{figure}

The same idea is applicable to $Q$-hyperproperties:
\begin{itemize}
\item The $\Sigma_n$-hyperproperties are the $Q$-hyperproperties such that $Q$ is a $\Sigma_n$-prefix and the defining relation $R$ is decidable.
\item The $\Pi_n$-hyperproperties are the $Q$-hyperproperties such that $Q$ is a $\Pi_n$-prefix and the defining relation $R$ is decidable.
\end{itemize}
Figure~\ref{fig:hierarchy} depicts this hierarchy.
Simply by reading off the quantifier prefix, any HyperLTL formula makes it easy to determine (an upper bound on) the hierarchy level in which it dwells.
Observational determinism (whose prefix is $\AAx$) is a $\Pi_1$-hyperproperty, as are all $k$-hyperproperties.  
Noninference (prefix $\AEx$) is a $\Pi_2$-hyperproperty, as are separability and generalized noninterference (prefix $\AAEx$).  
Their defining relations are decidable, because HyperLTL$_2$ validity is decidable. 

This arithmetic hierarchy of hyperproperties yields insight into verification.
Our model-checking algorithm in \S\ref{sec:modelchecking} permits up to one quantifier alternation, thus verifying a linear-time subclass of $\Pi_2$-hyperproperties.  
What about hyperproperties higher than $\Pi_2$ in the hierarchy? 
We don't yet know of any security policies that are examples.
As Rogers~\cite{Rogers87} writes, ``The human mind seems limited in its ability to understand and visualize beyond four or five alternations of quantifier.  Indeed, it can be argued that the inventions{\ldots}of mathematics are devices for assisting the mind in dealing with one or two additional alternations of quantifier.''
For practical purposes, we might not need to go much higher than $\Pi_2$.

\section{Related Work}
\label{sec:relatedwork}

McLean~\cite{McLean:1994:GeneralTheory} formalizes security policies as  closure with respect to \emph{selective interleaving functions}. 
He shows that trace properties cannot express security policies such as noninterference and average response time, because those are not properties of single execution traces.
Mantel~\cite{Mantel00} formalizes security policies with \emph{basic security predicates}, which stipulate \emph{closure conditions} for trace sets.

Clarkson and Schneider~\cite{ClarksonS10} introduce \emph{hyperproperties}, a framework for expressing security policies. 
Hyperproperties are sets of trace sets, and are able to formalize security properties such as noninterference, generalized noninterference, observational determinism and average response time. 
Clarkson and Schneider use second-order logic to formulate hyperproperties.  
That logic isn't verifiable, in general, because it cannot be effectively and completely axiomatized. 
Fragments of it, such as HyperLTL, can be verified.

van der Meyden and Zhang~\cite{VanDerMeyden:2007:verifOfNonInterf} use model-checking to verify  noninterference policies.
They reduce noninterference properties to safety properties expressible in standard linear and branching time logics. 
Their methodology requires customized model-checking algorithms for each security policy, whereas HyperLTL uses the same algorithm for every policy.

Dimitrova et al.~\cite{Dimitrova:2012:SecLTL} propose SecLTL, which extends LTL with a \emph{hide} modality $\mathcal{H}$ that requires observable behavior to be independent of secret values.
SecLTL is designed for output-deterministic systems.
Generalized noninterference~\eqref{hp:gni}, and other policies for nondeterministic systems, do not seem to be expressible with $\mathcal{H}$.

Balliu et al.~\cite{Balliu::Epis} use a linear-time temporal epistemic logic to specify many declassification policies derived from noninterference.
Their definition of noninterference, however, seems to be that of observational determinism~\eqref{hp:od}.  
They do not consider any information-flow policies involving existential quantification, such as noninference~\eqref{hp:ni}. 
They also do not consider systems that accept inputs after execution has begun.
Halpern and O'Neill~\cite{HalpernO08} use a similar temporal epistemic logic to specify \emph{secrecy} policies, which subsume many definitions of noninterference; they do not pursue model checking algorithms. 

Milishev and Clarke~\cite{MilushevC12,Milushev:2013:thesis} propose a verification methodology based on formulating hyperproperties as coinductive predicates over trees.  
They use the polyadic modal $\mu$-cal\-cu\-lus~\cite{Andersen94} to express hyperproperties and \emph{game-based} model-checking to verify them. 
Their logic, because it includes fixpoint operators, seems to be more expressive than HyperLTL.  
Nonetheless, HyperLTL is able to express many security policies, suggesting that a simpler logic suffices.

\section*{Acknowledgements}
Fred B. Schneider suggested the name ``HyperLTL'' for our logic.
We thank him, Dexter Kozen, Jos\'e Meseguer, and Moshe Vardi for discussions about this work.
Adam Hinz worked on an early prototype of the model checker.
This work was supported in part by AFOSR grant FA9550-12-1-0334 and NSF grant CNS-1064997.

\bibliographystyle{plain}
\bibliography{references}

\appendix


\section{The insufficiency of branching-time logic}
\label{sec:whynotctl}

CTL and CTL$^*$ have explicit path quantifiers.  
It's tempting to try to express security policies with those quantifiers.
Unfortunately, that doesn't work for information-flow policies such as observational determinism~\eqref{hp:od}.
Consider the following fragment of CTL$^*$ semantics~\cite{EmersonH86}:
\begin{equation*}
\begin{array}{ll}
s \models \A \phi &~\text{iff for all $\pi \in M$, if $\pi(1) = s$ then $\pi \models \phi$} \\
\pi \models \Phi &~\text{iff $\pi(1) \models \Phi$}
\end{array}
\end{equation*}
\emph{Path formulas} $\phi$ are modeled by paths $\pi$, and \emph{state} formulas $\Phi$ are modeled by states $s$.
Set $M$ is all paths in the model.
State formula $\A \phi$ holds at state $s$ when all paths proceeding from $s$ satisfy $\phi$.
Any state formula $\Phi$ can be treated as a path formula, in which case $\Phi$ holds of the path iff $\Phi$ hold of the first state on that path.
Using this semantics, consider the meaning of $\AAx \phi$, which is the form of observational determinism~\eqref{hp:od}:
\begin{equation*}
\begin{array}{l}
s \models \AAx \phi \\
= \text{for all $\pi \in M$ if $\pi(1)=s$ then $\pi \models \A \phi$} \\
= \text{for all $\pi \in M$ and $\pi' \in M$, if $\pi(1)=\pi'(1)=s$ then $\pi' \models \phi$} 
\end{array}
\end{equation*}
Note how the meaning of $\AAx \phi$ is ultimately determined by the meaning of $\phi$, where $\phi$ is modeled by the single trace $\pi'$.  
Trace $\pi$ is ignored in determining the meaning of $\phi$;
the second universal path quantifier causes $\pi$ to ``leave scope.''
Hence $\phi$ cannot express correlations between $\pi$ and $\pi'$, as observational determinism requires.
So CTL$^*$ path quantifiers do not suffice to express information-flow policies.
Neither do CTL path quantifiers, because CTL is a sub-logic of CTL$^*$.
Self-composition does enable expression of some, though not all, information-flow policies in branching-time logics, as we discuss in \S\ref{sec:HLTLvsOthers}.

\section{Model-checking Constructions}
\label{sec:constructions}

\subsection{Self-composition construction}
\label{sec:selfcomp}

Self-composition is the technique that Barthe et al.~\cite{BartheDR04} adopt to verify noninterference policies.  It was generalized by Terauchi and Aiken~\cite{TerauchiA05} to verify observational determinism policies~\cite{Zdancewic:2005:policyDown,ZdancewicMyers:2003:OD},
and by Clarkson and Schneider~\cite{ClarksonS10} to verify $k$-safety hyperproperties.
We  extend this technique to model-checking of {\Hp}.

\paragraph*{{\Buchi} automata.} {\Buchi} automata~\cite{Vardi:1994:infComp} are finite-state automata that accept strings of infinite length.
A {\Buchi} automaton is a tuple $(\Sigma,S,\Delta,S_0,F)$ where $\Sigma$ is an alphabet, $S$ is the set of states, $\Delta$ is the transition relation such that $\Delta\subseteq S\times\Sigma\times S$, $S_0$ is the set of initial states, and $F$ is the set of accepting states, where both $S_0\subseteq S$ and $F\subseteq S$.
A \emph{string} is a sequence of letters in $\Sigma$.
A path $s_0s_1\dots$ of a {\Buchi} automaton is \emph{over} a string $\alpha_1\alpha_2\dots$ if, for all $i\geq 0$, it holds that $(s_i,\alpha_{i+1},s_{i+1})\in\Delta$.
A string is \emph{recognized} by a {\Buchi} automaton if there exists a path $\pi$ over the string with some accepting states occurring infinitely often, in which case $\pi$ is an \emph{accepting path}.
The \emph{language} $\mathcal{L}(A)$ of an automaton $A$ is the set of strings that automaton accepts.
A {\Buchi} automaton can be derived~\cite{Clarke:1999:model-checking} from a Kripke structure, which   is a common mathematical model of interactive, state-based systems.

\paragraph*{Self composition.}
The \emph{$n$-fold self-composition} $A^n$ of {\Buchi} automaton $A$ is essentially the product of $A$ with itself, $n$ times. This construction is defined as follows:
\begin{definition}
{\Buchi} automaton $A^n$ is the \emph{$n$-fold self-composition} of {\Buchi} automaton $A$, where $A=(\Sigma,S,\Delta,S_0,F)$, if $A^n=(\Sigma^n,S^n,\Delta',S^n_0,F^n)$ and for all $s_1,s_2\in S^n$ and $\alpha\in\Sigma^n$ we have $(s_1,\alpha,s_2)\in\Delta' $ iff
for all $1\leq i\leq n$, it holds that $(\proj_i(s),\proj_i(\alpha),\proj_i(s'))\in\Delta$.
\end{definition}
\noindent $A^n$ recognizes $zip(\gamma_1, \dots, \gamma_n)$ if $A$ recognizes each of $\gamma_1,\dots,\gamma_n$:
\begin{proposition}
$\mathcal{L}(A^n) = \{\zip(\gamma_1,\dots,\gamma_n) \mid \gamma_1,\dots,\gamma_n\in\mathcal{L}(A)\}$
\end{proposition}

\begin{proof}
By the construction of $A^n$.
\end{proof}

\subsection{Formula-to-automaton construction}
\label{sec:automataCons}

Given a {\Hp} formula $\Ak\Ej\psi$ and a set $B$ of bonding functions, we now show how to construct an automaton that accepts exactly the strings $w$ for which $\unzip(w)\models\psi$. 
Our construction extends standard methodologies for  LTL automata construction~\cite{GerthVardi:1995:OntheFlyLTL,Vardi:2007:automata-theoretic-rev,Vardi:1996:LTL,Gastin:2001:FastLTLtoBuchi}.

\paragraph{1. Negation normal form.}
We begin by preprocessing $\psi$ to put it in a form more amenable to model checking.
The formula is rewritten to be in \emph{negation normal form} (NNF), meaning (i)  negation connectives are applied only to atomic propositions in $\psi$, (ii) the only connectives used in $\psi$ are \X{}, \U{}, \R{}, $\neg$, $\vee$, $\wedge$, and focus formulas, and (iii)
every focus formula contains exactly one non-$\top$ subformula---for example, $\<\top,\ldots,\top,p,\top,\ldots,\top\>$ or $\<\top,\ldots,\top,\neg p,\top,\ldots,\top\>$---and that subformula must be an atomic proposition or its negation.
We identify $\neg\neg\psi$ with $\psi$.

\paragraph{2. Construction.} 
We now construct a \emph{generalized {\Buchi} automaton}~\cite{CourcoubVardi:1992:MemEffAlgo} $A_{\psi}$ for $\psi$. 
A generalized {\Buchi} automaton is the same as a {\Buchi} automaton except that it has multiple sets of accepting states. 
That is, a generalized {\Buchi} automaton is a tuple $(\Sigma,S,\Delta,S_0,F)$ where $\Sigma$, $S$, $\Delta$ and $S_0$ are defined as for {\Buchi} automata, and $F=\setdef{F_i}{1\leq i\leq m \text{ and } F_i\subseteq S}$.
Each of the $F_i$ is an \emph{accepting set}.
A string is recognized by a generalized {\Buchi} automaton if there is a path over the string with at least one of the states in every accepting set occurring infinitely often.

To construct the states of $A_\psi$, we need some additional definitions. 
Define \emph{closure} $\cl(\psi)$ of $\psi$ to be the least set of subformulas of $\psi$ that is closed under the following rules:
\begin{itemize}
	\item if $\<\top,\ldots,\top,\psi',\top,\ldots,\top\>\in \cl(\psi)$, then $\<\top,\ldots,\top,\neg \psi',\top,\ldots,\top\>\in \cl(\psi)$.
	\item if $\psi'\in \cl(\psi)$ and $\psi'$ is not in the form of $\<\psi_1,\dots,\psi_n\>$, then $\neg\psi'\in \cl(\psi)$.
	\item if $\psi_1\wedge \psi_2\in \cl(\psi)$ or $\psi_1\vee \psi_2\in \cl(\psi)$, then $\{\psi_1,\psi_2\}\subseteq \cl(\psi)$.
	\item if $\X \psi'\in \cl(\psi)$, then $\psi'\in \cl(\psi)$.
	\item if $\psi_1\U \psi_2\in \cl(\psi)$ or $\psi_1\R \psi_2\in \cl(\psi)$, then $\{\psi_1,\psi_2\}\subseteq \cl(\psi)$.
\end{itemize} 
And define $K$ to be a \emph{maximal consistent set} with respect to $\cl(\psi)$ if $K\subseteq\cl(\psi)$ and the following conditions hold:
\begin{itemize}
	\item if $\psi'$ is not a focus formula, then ($\psi'\in K$ iff $\neg \psi'\not\in K$).
	\item if $\psi'$ is a focus formula $\<\top,\ldots,\top,\psi',\top,\ldots,\top\>$ and $\psi' \in \cl(\psi)$, then ($\psi'\in K$ iff $\<\top,\ldots,\top,\neg \psi',\top,\ldots,\top\> \not\in K$).
	\item if $\psi_1\wedge \psi_2\in \cl(\psi)$, then ($\psi_1\wedge \psi_2\in K$ iff $\{\psi_1,\psi_2\} \subseteq K$).
	\item if $\psi_1\vee \psi_2\in \cl(\psi)$, then ($\psi_1\vee \psi_2\in K$ iff $\psi_1\in K$ or $\psi_2\in K$).
	\item if $\psi_1\U \psi_2\in K$ then $\psi_1\in K$ or $\psi_2\in K$.
	\item if $\psi_1\R \psi_2\in K$ then $\psi_2\in K$.
\end{itemize} 

\noindent Define $\ms(\psi)$ to be the set of all maximal consistent sets with respect to $\psi$. 
The elements of $\ms(\psi)$ will be the states of $A_\psi$; hence each state is a set of formulas.
Intuitively, a state $s$ describes a set of computation tuples where each tuple is a model of all the formulas in $s$.  
There will be a transition from a state $s_1$ to a state $s_2$ iff every computation tuple described by $s_2$ is an \emph{immediate suffix} of some tuple described by $s_1$.
(Tuple $\Gamma$ is an immediate suffix of $\Gamma'$ iff $\Gamma = \Gamma'[2..]$.)

Automaton $A_{\psi}=(\Sigma_{\psi},S_{\psi},\Delta_{\psi},\{\iota_{\psi}\},F_{\psi})$ is defined as follows:
\begin{itemize}
	\item The alphabet $\Sigma_{\psi}$ is $\mathcal{P}(\Atoms)^n$.  Each letter of the alphabet is, therefore, an $n$-tuple of sets of atomic propositions.
	\item The set $S_{\psi}$ of states is $\ms(\psi) \cup \{\iota_\psi\}$, where $\ms(\psi)$ is defined above and $\iota_{\psi}$ is a distinct initial state.
	\item The transition relation $\Delta_\psi$ contains $(s_1,\alpha,s_2)$, where  $\{s_1,s_2\} \subseteq S_{\psi}\setminus\{\iota_\psi\}$ and $\alpha\in\Sigma_{\psi}$, iff
		\begin{itemize}
			\item For all $p\in \Atoms$, if $\<\top,\ldots,\top,p,\top,\ldots,\top\> \in s_2$, and $p$ is element $i$ of that focus formula, then $p\in \proj_i(\alpha)$.
				Likewise, if $\<\top,\ldots,\top,\neg p,\top,\ldots,\top\> \in s_2$, then $p\not\in \proj_i(\alpha)$.
			\item For all $p\in \Compounds$, if $p\in s_2$ then $p\in B(\alpha).$ Likewise, if $\neg p\in s_2$ then $p\not\in B(\alpha).$
			\item If $\X\psi'\in s_1$ then $\psi'\in s_2$. 
			\item If $\psi_1\U\psi_2\in s_1$ and $\psi_2\not\in s_1$ then $\psi_1\U\psi_2\in s_2$. 
			\item If $\psi_1\R\psi_2\in s_1$ and $\neg\psi_1\in s_1$ then $\psi_1\R\psi_2\in s_2$. 
		\end{itemize}
  And $\Delta_\psi$ contains $(\iota_{\psi},\alpha,s_2)$ iff $\psi\in s_2$ and $(\iota_{\psi},\alpha,s_2)$ is a transition permitted by the above rules for $\Atoms$ and $\Compounds$. 
    \item The set of initial states contains only $\iota_\psi$.
	\item The set $F_\psi$ of sets of accepting states contains one set $\setdef{s\in (S_{\psi}\setminus\{\iota_\psi\})}{\neg(\psi_1\U\psi_2)\in s \text{ or }\psi_2\in s}$ for each until formula $\psi_1\U\psi_2$ in $\cl(\psi)$.
\end{itemize}
The definition of $F_\psi$ guarantees that, for every until formula $\psi_1\U\psi_2$, eventually $\psi_2$ will hold.
	That is because the transition rules don't allow a transition from a state containing $\psi_1\U\psi_2$ to a state containing $\neg(\psi_1\U\psi_2)$ unless $\psi_2$ is already satisfied.

\paragraph{3. Degeneralization of {\Buchi} automata.} Finally, convert generalized {\Buchi} automaton $A_{\psi}$ to a ``plain'' {\Buchi} automaton. This conversion is entirely standard~\cite{GerthVardi:1995:OntheFlyLTL}, so we do not repeat it here.

\paragraph*{Correctness of the construction.}
The following proposition states that $A_\psi$ is constructed such that it recognizes computation tuples that model $\psi$:
\begin{proposition}
 $\Gamma\models\psi$ iff $\zip(\Gamma)\in\mathcal{L}(A_{\psi})$.
\end{proposition}
\begin{proof}
($\Leftarrow$)
By the construction of $A_{\psi}$, the states with a transition from $\iota_{\psi}$ contain $\psi$. Hence by Lemma~\ref{lem:autoStrings} below, for all the strings $w$ such that $w=\zip(\Gamma)$ in $\mathcal{L}(A_{\psi})$, it holds that $\Gamma\models\psi$.
\par
($\Rightarrow$) 
Let $s_i=\{\psi'\in\cl(\psi)\mid \Gamma[i..]\models\psi'\}$ for all $i\in\mathbb{N}$. Then by the definition, $s_i\in\ms(\psi)$. We show that $\iota_{\psi}s_1s_2\dots$ is an accepting path in $A_{\psi}$. By $\Gamma\models\psi$ we have $\psi\in s_1$. By the construction of $A_{\psi}$, $(i_{\psi},\alpha_1,s_1)\in\Delta_{\psi}$ where $\alpha_1=\zip(\Gamma)[1]$. The construction of the path inductively follows the construction of $A_{\psi}$, which respects the semantics of HyperLTL.
\end{proof}

\begin{lemma}
\label{lem:autoStrings}
Let $\iota_{\psi}s_1\dots$ be an accepting path in $A_{\psi}$ over the string $w=\alpha_1\alpha_2\dots$. Let $\Gamma=\unzip(w)$. Then for all $i\geq 0$, it holds that $\psi'\in s_i$ iff $\Gamma[i..]\models\psi'$.
\end{lemma}
\begin{proof}
The proof proceeds by induction on the structure of $\psi'$:
\smallskip

\noindent\textbf{Base cases:}
\begin{enumerate}
	\item $\psi'=p$ where $p\in\Compounds$
	
	($\Rightarrow$) Assume that $p\in s_i$. By the construction of $A_{\psi}$, if $p\in s_i$ then $p\in B_n(\alpha_i)$ or equivalently $p\in B_n(\Gamma[i])$. By the semantics of HyperLTL, we have $\Gamma[i..]\models p$.
	
	($\Leftarrow$) Assume that $\Gamma[i..]\models p$.
	Then $p\in B_n(\Gamma[i])$, which is equivalent to $p\in B_n(\alpha_i)$. By the fact that states are maximal consistent sets, one of $p$ or $\neg p$ must appear in $s_i$.
	By the construction of $A_{\psi}$ and the fact that $p\in B_n(\alpha_i)$, we have $p\in s_i$.
	
	\item $\psi'=\langle\psi_1,\dots,\psi_n\rangle$
	
	($\Rightarrow$) Assume that $\psi'\in s_i$.
	By the construction of $A_{\psi}$, for all $1\leq r\leq n$, if $\psi_r=p$ we have $p\in \proj_r(\alpha_i)$ or equivalently $p\in\proj_r(\Gamma[i])$.
	Then by the semantics of HyperLTL, $\proj_r(\Gamma[i])\models\psi_r$.
	If $\psi_r=\neg p$, then $p\not\in\proj_r(\Gamma[i])$, which again concludes $\proj_r(\Gamma[i])\models\psi_r$.
	Therefore we have $\Gamma[i..]\models\psi'$.
	
	($\Leftarrow$) Assume that $\Gamma[i..]\models\psi'$.
	Then for all $1\leq r\leq n$, $\proj_r(\Gamma[i])\models\psi_r$.
	If $\psi_r=p$ we have $p\in\proj_r(\Gamma[i])$, which is equivalent to $p\in\proj_r(\alpha_i)$.
	If $\psi_r=\neg p$ then $p\not\in\proj_r(\Gamma[i])$, which is $p\not\in\proj_r(\alpha_i)$.
	By the fact that states are maximal consistent sets, one of $\langle\psi_1,\dots,\psi_n\rangle$ or a member of $\overline{\<\psi_1,\dots,\psi_n\>}$ must appear in $s_i$.
	By the semantics of HyperLTL and the construction of $A_{\psi}$, only $\langle\psi_1,\dots,\psi_n\rangle$ can be in $s_i$, that means $\psi'\in s_i$.
\end{enumerate}

\noindent\textbf{Inductive cases:}
\begin{enumerate}
	\item $\psi'=\neg\psi''$
	
	($\Rightarrow$) Assume that $\neg\psi''\in s_i$, Then $\psi''\not\in s_i$. By induction hypothesis, $\Gamma[i..]\not\models\psi''$, or equivalently, $\Gamma[i..]\models\neg\psi''$. Hence, $\Gamma[i..]\models\psi'$.

	($\Leftarrow$) Similar to $\Rightarrow$.
	
	\item $\psi'=\psi_1\vee\psi_2$
	
	($\Rightarrow$) By the construction of $A_{\psi}$, if $\psi_1\vee\psi_2\in s_i$ then $\psi_1\in s_i$ or $\psi_2\in s_i$. By induction hypothesis, $\Gamma[i..]\models\psi_1$ or $\Gamma[i..]\models\psi_2$, which concludes $\Gamma[i..]\models\psi_1\vee\psi_2$.
	
	($\Leftarrow$) Similar to $\Rightarrow$.
	
	\item $\psi'=\X\psi''$ 
	
	($\Rightarrow$) Assume that $\psi'\in s_i$. By the  construction of $A_{\psi}$, $\psi''\in s_{i+1}$. By induction hypothesis, $\Gamma[i+1..]\models\psi''$, which concludes $\Gamma[i..]\models\X\psi''$.
	
	($\Leftarrow$) Similar to $\Rightarrow$ and the fact that always one of $\X\psi''$ or $\neg\X\psi''$ appears in a state.
	
	\item $\psi'=\psi_1\U\psi_2$
	
	($\Rightarrow$) Assume that $\psi_1\U\psi_2\in s_i$. By the construction of $A_{\psi}$ and the fact that the path is accepting, there is some $j\geq i$ such that $\psi_2\in s_j$. Let $j$ be the smallest index. By induction hypothesis, $\Gamma[j..]\models\psi_2$. By the construction of $A_{\psi}$, for all $i\leq k < j$, $\psi_1\in s_k$. Therefore by induction hypothesis, $\Gamma[k..]\models\psi_1$. which concludes $\Gamma[i..]\models\psi_1\U\psi_2$.

	($\Leftarrow$) Similar to $\Rightarrow$.
	
	\end{enumerate}

\end{proof}

\end{document}